\documentclass[letterpaper,10pt,conference]{ieeeconf}  

\IEEEoverridecommandlockouts                              

\usepackage{graphics} 
\usepackage{epsfig} 
\usepackage{amsmath} 
\usepackage{amssymb}  
 
\usepackage{amsthm}
\usepackage{cite}
\usepackage{color}
\usepackage{enumerate}
\usepackage{graphicx}
\usepackage{mathtools}
\usepackage{cancel}
\usepackage{url}
\usepackage{algorithm}
\usepackage{algpseudocode}
\usepackage{textgreek}

\graphicspath{{figures/}}

\newcommand{\naturals}{\mathbb{N}}
\newcommand{\real}{\mathbb{R}}
\newcommand{\realnonneg}{\mathbb{R}_{\ge 0}}

\newcommand{\map}[3]{#1:#2 \rightarrow #3}

\newcommand{\longthmtitle}[1]{\mbox{}{\textit{(#1):}}}

\newcommand{\setdefb}[2]{\big\{#1 \; | \; #2\big\}}
\newcommand{\setdefB}[2]{\Big\{#1 \; | \; #2\Big\}}
\newcommand*{\SetSuchThat}[1][]{} 
\newcommand*{\MvertSets}{%
    \renewcommand*\SetSuchThat[1][]{%
        \mathclose{}%
        \nonscript\;##1\vert\penalty\relpenalty\nonscript\;%
        \mathopen{}%
    }%
}
\MvertSets 
\DeclarePairedDelimiterX \Set [2] {\lbrace}{\rbrace}
    {\,#1\SetSuchThat[\delimsize]#2\,}

\newcommand{\Cc}{\mathcal{C}}

\newcommand{\Kc}{\mathcal{K}}

\newcommand{\Tc}{\mathcal{T}}

\newcommand{\Sc}{\mathcal{S}}

\newcommand{\R}{\mathbb{R}}

\newcommand{\Ke}{\Kc^{\rm e}}

\newtheorem{theorem}{Theorem}
\newtheorem{lemma}{Lemma}
\newtheorem{corollary}{Corollary}
\newtheorem{proposition}{Proposition}

\theoremstyle{definition}
\newtheorem{definition}{Definition}

\newtheorem{remark}{Remark}

\newcommand{\bb}{\mathbf{b}}

\renewcommand{\bf}{\mathbf{f}} 
\newcommand{\bg}{\mathbf{g}}

\newcommand{\bk}{\mathbf{k}}

\newcommand{\bu}{\mathbf{u}}
\newcommand{\bv}{\mathbf{v}}

\newcommand{\bx}{\mathbf{x}}

\newcommand{\bF}{\mathbf{F}}

\newcommand{\bI}{\mathbf{I}}

\newcommand{\bGamma}{\boldsymbol{\Gamma}}
\newcommand{\bzero}{\mathbf{0}}

\newcommand{\des}{{\operatorname{d}}}

\newcommand{\relu}{{\operatorname{ReLU}}}

\renewcommand{\baselinestretch}{1}
\allowdisplaybreaks

\DeclareMathOperator*{\argmin}{argmin}

\title{\textbf{On the Properties of Optimal-Decay Control Barrier Functions}}

\author{Pio Ong, Max H. Cohen, Tamas G. Molnar, Aaron D. Ames %
\thanks{PO, MC and AA are with the Department of Mechanical and Civil Engineering, California Institute of Technology, Pasadena, CA \texttt{\{pioong,maxcohen,ames\}@caltech.edu}.}
\thanks{TM is with the Department of Mechanical Engineering, Wichita State University, Wichita, KS \texttt{\{tamas.molnar\}@wichita.edu}.}
\thanks{This research was supported by the Technology Innovation Institute (TII), Boeing, and NSF CPS Award \#1932091.}
}
\begin{document}

\maketitle
\begin{abstract}
Control barrier functions provide a powerful means for synthesizing safety filters that ensure safety framed as forward set invariance.  Key to CBFs' effectiveness is the simple inequality on the system dynamics: $\dot{h} \geq - \alpha(h)$.  Yet determining the class $\mathcal{K}^e$ function $\alpha$ is a user defined choice that can have a dramatic effect on the resulting system behavior. This paper formalizes the process of choosing $\alpha$ using {\em optimal-decay control barrier functions} (OD-CBFs). These modify the traditional CBF inequality to: $\dot{h} \geq - \omega \alpha(h)$, where $\omega \geq 0$ is automatically determined by the safety filter. A comprehensive characterization of this framework is elaborated, including tractable conditions on OD-CBF validity, control invariance of the underlying sets in the state space, forward invariance conditions for safe sets, and discussion on optimization-based safe controllers in terms of their feasibility, Lipschitz continuity, and closed-form expressions.   The framework also extends existing higher-order CBF techniques, addressing safety constraints with vanishing relative degrees. The proposed method is demonstrated on a satellite control problem in simulation.
\end{abstract}

\section{Introduction}

Control barrier functions (CBFs), first conceptualized in \cite{WielandSNCS07} and later refined in
\cite{AmesTAC17}, provide a framework for safety-critical control of nonlinear systems through the concept of set invariance \cite{Blanchini}. A core contribution of the modern version of CBFs introduced in \cite{AmesTAC17} over previous versions \cite{WielandSNCS07}, and other instantiations of barrier functions (BFs) \cite{PrajnaHSCC05}, is the addition of an extended class $\mathcal{K}$ function, denoted by $\mathcal{K}^e$, to the Lyapunov-like inequality characterizing safety. This modification has important practical ramifications: rather than every superlevel set, only the zero superlevel set of the BF
is rendered forward invariant; this set is also, under certain assumptions, asymptotically stable \cite{AmesADHS15,TeelTAC24}; and the set of safe controllers meeting the corresponding CBF requirement is greatly expanded, facilitating the use of optimization-based control approaches, typically via quadratic programs (QPs). This modern notion of a BF also closes the gap between the necessary and sufficient conditions for set invariance originally stated by Nagumo \cite{MN:42} (see \cite[Ch. 4.2]{Blanchini} and \cite[Ch. 4.1]{AbrahamMarsdenRatiu} for a modern account) in that compact forward invariant sets necessarily admit a BF satisfying the Lyapunov-like conditions in \cite{AmesTAC17}.

Although the modern version of CBFs \cite{AmesTAC17} has expanded the practical applications of safety-critical control techniques, this class $\mathcal{K}^e$ function characterization of set invariance also brings certain drawbacks. Namely, one must now not only search for a control invariant set -- the zero superlevel set of the corresponding CBF -- contained within a desirable set of system states (that is a challenging problem on its own), but also for a class $\mathcal{K}^e$ function compatible with the dynamics and invariant set. When the CBF has a uniform relative degree and the set of admissible inputs is unbounded, any class $\mathcal{K}^e$ function will be compatible with the dynamics and CBF. However, if there exist states where the (candidate) CBF's relative degree vanishes or the set of admissible inputs is bounded, the choice of class $\mathcal{K}^e$ function significantly impacts the validity of this candidate CBF. 

The challenges associated with searching for a compatible class $\mathcal{K}^e$ function in the modern CBF framework are further amplified when using higher-order CBF techniques, such as exponential CBFs (ECBFs) \cite{QN-KS:16} and high-order CBFs (HOCBFs) \cite{WeiTAC22,TanTAC22}. Such techniques dynamically extend a given safety constraint that does not meet the criteria of a CBF to a new function that, under certain conditions, is a valid certificate for set invariance and satisfies a Lyapunov-like inequality characterized by a class $\mathcal{K}^e$ function. Yet, as illustrated in, e.g., \cite[Example 1]{PO-MHC-TGM-ADA:25-csl} and \cite[Example 5]{CohenARC24}, there often exists no class $\mathcal{K}^e$ function compatible with the dynamics and corresponding ECBF/HOCBF if there exist points within the safe set where the relative degree of such an ECBF/HOCBF vanishes, even if the admissible set of inputs is unbounded. That is, even for seemingly benign dynamics and safety constraints, ECBFs and HOCBFs, as defined in \cite{QN-KS:16,WeiTAC22,TanTAC22}, may not exist. Alternative higher-order CBF techniques, such as rectified CBFs (ReCBFs) \cite{PO-MHC-TGM-ADA:25-csl} and backstepping CBFs \cite{AndrewCDC22,CohenLCSS24}, maintain compatibility of the class $\mathcal{K}^e$ function under vanishing relative degree by construction, but this phenomenon can often degrade performance by producing controllers with large Lipschitz constants. Indeed, as studied in \cite{BrunkeArXiV24}, sampled-data implementations of CBF-based controllers with vanishing relative degree can produce chattering despite the fact that, under certain assumptions \cite{MJ:18}, such inputs are guaranteed to be locally Lipschitz when implemented in continuous-time.

One technique that has the potential to address some of these aforementioned limitations is to scale the class $\mathcal{K}^e$ function by a tunable coefficient whose value is optimized, in real-time, alongside the control input. Specifically, CBF-based control inputs are typically computed as the solution to a QP, and by adding a tunable coefficient as a decision variable to this QP, one effectively gains an additional control input that provides more flexibility to satisfy the Lyapunov-like conditions for set invariance. This approach is used often among practitioners, but there is relatively little information in the formal CBF literature concerning the properties of this approach, with exceptions being  \cite{AmesIEEEA20} and \cite{ZengACC21}. In particular, \cite{AmesIEEEA20} is, to our knowledge, the first to explicitly acknowledge that one may optimize such a coefficient within the CBF-QP, but does not establish any properties of the resulting controller. The most in-depth study of this technique is presented in \cite{ZengACC21} in the context of ensuring feasibility for CBF-QPs in the presence of bounded inputs, where the corresponding controller is referred to as an \emph{optimal decay} CBF-QP (OD-CBF-QP). Importantly,  \cite{ZengACC21} establishes that this OD-CBF-QP is feasible for any point on the interior of the safe set, but does not make any statements regarding feasibility on the boundary (a necessary condition for safety), feasibility outside the safe set (useful for ensuring stability of the safe set), or Lipschitz continuity of the resulting controller (needed to guarantee existence and uniqueness of solutions to the closed-loop system).

In this paper, we build upon the insights of \cite{AmesIEEEA20, ZengACC21} to provide a rigorous characterization of OD-CBFs and their properties. We begin with motivating our developments by introducing a generalized notion of a BF, an optimal decay BF (OD-BF). The main advantage of OD-BFs over those in \cite{AmesTAC17} is that OD-BFs provide tighter connections to Nagumo's conditions for set invariance: OD-BFs are necessary and sufficient for set invariance without any compactness assumptions. With this as motivation, we then revisit OD-CBFs as introduced in \cite{AmesIEEEA20, ZengACC21}. Here, we provide tractable conditions to verify that a given function is an OD-CBF, establish the feasibility of the corresponding OD-CBF-QP at all points in the safe set, illustrate that the resulting controller is locally Lipschitz continuous, and that it renders the safe set forward invariant. Under further assumptions, we show that the OD-CBF-QP is also feasible outside the safe set and renders the safe set asymptotically stable if it is compact. We then extend these developments to higher-order CBF techniques to address the aforementioned issues with a vanishing relative degree, and highlight via numerical examples the advantage of the OD-CBF framework over traditional CBF approaches. Finally, we showcase the developed framework on a more complicated satellite control problem. 

\section{Preliminaries}

\textbf{Forward Invariance:}
Consider a the nonlinear system\footnote{
We denote the sets of natural, real, and nonnegative real numbers by $\naturals$, $\real$, and $\realnonneg$. Given ${\bu\in\real^m}$, ${\|\bu\| = \sqrt{\bu^\top \bu}}$ is its Euclidean norm, while ${\|\bu\|_{\bGamma} = \sqrt{\bu^\top \bGamma \bu}}$ with a positive definite matrix ${\bGamma \in \real^{m \times m}}$. For continuously differentiable functions
${\map{h}{\real^n}{\real}}$,
${\map{\bf}{\real^n}{\real^n}}$,
${\map{\bg}{\real^n}{\real^{n\times m}}}$,
we use the Lie derivatives
${L_{\bf}h(\bx) = \frac{\partial h}{\partial \bx}(\bx) \cdot \bf(\bx)}$ and
${L_{\bg}h(\bx) = \frac{\partial h}{\partial \bx}(\bx) \cdot \bg(\bx)}$.
Function ${\map{\alpha}{\realnonneg}{\real}}$ is of class-$\Kc$ (${\alpha\in\Kc}$) if it is continuous, strictly increasing, and satisfies ${\alpha(0)=0}$.
If the same properties hold for a function ${\map{\alpha}{\real}{\real}}$, then it is of extended class-$\Kc$ (${\alpha\in\Ke}$).
$\relu$ denotes the rectified linear unit, i.e., ${\relu(\cdot) = \max\{0,\cdot\}}$.}:
\begin{equation}
    \label{sys:autonomous}
    \dot \bx = \bF(\bx),
\end{equation}
with state ${\bx\in\real^n}$ and dynamics $\map{\bF}{\real^n}{\real^n}$. One desired system property is \textit{safety}, where state trajectories $\bx(t)$ must remain in a safety constraint set defined by ${\map{\psi}{\real^n}{\real}}$:
\begin{equation}\label{eq:safety_constraint}
    \Sc = \setdefb{\bx\in\real^n}{\psi(\bx)\geq 0}.
\end{equation}
In other words, safety requires ${\bx(t)\in\Sc}$ for all time. A related concept to this is set \textit{forward invariance}.
\begin{definition}\longthmtitle{Forward Invariance} A set ${\Cc \subset \real^n}$ is \textbf{forward invariant} for system~\eqref{sys:autonomous} if, for any initial condition ${\bx_0\in\Cc}$, the trajectory $\bx(t)$ is contained within $\Cc$, i.e., ${\bx(t)\in\Cc}$ for all time $t\geq 0$. The set $\Cc$ is a \textbf{safe set} for~\eqref{sys:autonomous} if it is also a subset of the safety constraint set, i.e., $\Cc\subseteq \Sc$.
\end{definition}
To verify safety, we seek a forward invariant set $\Cc\subseteq\Sc$, ideally $\Sc$ itself, which ensures that the system can operate safely if it is initialized within the set $\Cc$.

We construct a forward invariant set $\Cc$ as the 0-superlevel set of a continuously differentiable function ${\map{h}{\real^n}{\real}}$:
\begin{equation}\label{eq:safeset}
\begin{aligned} 
    \Cc &= \setdefb{\bx\in\real^n}{h(\bx)\geq 0},\\
    \partial\Cc &= \setdefb{\bx\in\real^n}{h(\bx)= 0}. 
\end{aligned}
\end{equation}
Barrier functions certify that $\Cc$ is forward invariant.
\begin{definition}\longthmtitle{Barrier Function}
    The function $h$ from~\eqref{eq:safeset} is a \textbf{barrier function} (BF) for~\eqref{sys:autonomous} on~$\Cc$ if there exists $\alpha\in\Kc$ such that for all $\bx\in\Cc$:
    \begin{equation}\label{eq:BC}
    L_\bF h(\bx) \geq -\alpha(h(\bx)).
    \end{equation}
\end{definition}
\begin{lemma}\label{lem:BF-invariance}
    \!Let $h$ be a BF for~\eqref{sys:autonomous} on~$\Cc$.
    If $\bF$ is locally Lipschitz and $\frac{\partial h}{\partial \bx}(\bx) \neq \bzero$ when $h(\bx)=0$, then $\Cc$ is forward invariant.
\end{lemma}
Condition~\eqref{eq:BC} enforced by a BF is a sufficient condition for establishing forward invariance. A tighter condition is given in the following result. 
\begin{lemma}\label{lem:not-nagumo}
    Let the system dynamics $\bF$ for~\eqref{sys:autonomous} be locally Lipschitz and suppose that $\frac{\partial h}{\partial \bx}(\bx) \neq \bzero$ when $h(\bx)=0$. Then, $\Cc$ is forward invariant if and only if:
    \begin{equation}
        \label{eq:nagumo}
        h(\bx)=0 \implies L_\bF h(\bx)\geq 0. 
    \end{equation}
\end{lemma}
The above result is derived from Nagumo's theorem~\cite{MN:42} and suggests that the barrier condition~\eqref{eq:BC} only needs to hold where $h(\bx)=0$ without the need to find the class-$\Kc$ function $\alpha$. With an additional assumption, the gap between the two results can be closed.
\begin{proposition}\longthmtitle{Converse Barrier Function \cite{AmesTAC17}}\label{prop:converse_BF}
    Let the system dynamics $\bF$ for~\eqref{sys:autonomous} be locally Lipschitz and suppose that $\frac{\partial h}{\partial \bx}(\bx) \neq \bzero$ when $h(\bx)=0$. Additionally, assume that $h$ has compact superlevel sets. Then, $\Cc$ is forward invariant if and only if $h$ is a BF.
\end{proposition}
The latter condition~\eqref{eq:nagumo} is often more useful to verify safety, whereas the BF framework becomes more relevant in the context of control systems, which we discuss next.

\textbf{CBFs:}
Consider the control affine system:
\begin{equation} \label{sys:ctrl_affine}
    \dot \bx = \bf(\bx) + \bg(\bx)\bu
\end{equation}
with state ${\bx\in\real^n}$, control input ${\bu\in\real^m}$, locally Lipschitz drift ${\map{\bf}{\real^n}{\real^n}}$ and control matrix ${\map{\bg}{\real^n}{\real^{n\times m}}}$. We address the \textit{safety-critical control} problem, where the goal is to keep system trajectories within the safety constraint set, i.e., ensure ${\bx(t)\in\Sc}$ for all time. Here, rather than forward invariant sets, we are interested in control invariant sets. 
\begin{definition}\longthmtitle{Safe Set}
    A set $\Cc\subset \real^n$ is a \textbf{control invariant set} for system~\eqref{sys:ctrl_affine} if, for any initial condition $\bx_0\in\Cc$, there exists a control signal $t\mapsto \bu(t)$ such that the resulting state trajectory $\bx(t)$ is contained within $\Cc$, i.e., $\bx(t)\in\Cc$ for all time $t\geq 0$. The set $\Cc$ is a \textbf{safe set} for~\eqref{sys:ctrl_affine} if it is also a subset of the safety constraint set, i.e., $\Cc\subseteq\Sc$.
\end{definition}

CBFs extend BFs to control systems and are a simple yet powerful tool for certifying that the set $\Cc$ is control invariant.

\begin{definition}\longthmtitle{Control Barrier Function~\cite{AmesTAC17}}
The function $h$ from~\eqref{eq:safeset} is a \textbf{control barrier function} (CBF) for~\eqref{sys:ctrl_affine} on $\Cc$ if there exists ${\alpha \in \Ke}$ such that for all $\bx\in\Cc$: 
\begin{equation}\label{eq:CBC}
\sup_{\bu\in\real^m} L_{\bf}h(\bx) + L_{\bg}h(\bx) \bu > -\alpha(h(\bx)).
\end{equation}
\end{definition}

CBFs help establish controllers ${\map{\bk}{\real^n}{\real^m}}$, ${\bu = \bk(\bx)}$, that render $\Cc$ forward invariant for the closed-loop system: 
\begin{equation} \label{sys:closedloop}
    \dot \bx = \bf(\bx) + \bg(\bx)\bk(\bx) \triangleq \bF(\bx).
\end{equation}

\begin{theorem}[\!\!\cite{AmesTAC17}] \label{theo:CBF}
If $h$ is a CBF for~\eqref{sys:ctrl_affine} on $\Cc$, then any locally Lipschitz controller ${\map{\bk}{\real^n}{\real^m}}$ that satisfies: 
\begin{equation}
    L_{\bf}h(\bx) +L_{\bg}h(\bx) \bk(\bx) \geq - \alpha(h(\bx)),
\label{eq:safety_condition}
\end{equation}
for all ${\bx \in \Cc}$, renders the set $\Cc$ forward invariant for~\eqref{sys:closedloop}, i.e., $\bx(t)\in\Cc$ for all $t\geq 0$ if $\bx_0\in\Cc$.
\end{theorem} 
Theorem~\ref{theo:CBF} follows directly from applying Lemma~\ref{lem:BF-invariance} to the closed-loop system with the state feedback $\bu=\bk(\bx)$. Note, importantly, the subtle difference of the strict inequality in the condition~\eqref{eq:CBC} for CBFs and the nonstrict inequality in~\eqref{eq:BC} for BFs. The strictness serves two purposes. First, it guarantees that the regular value assumption, i.e., $\frac{\partial h}{\partial \bx}(\bx)\neq \bzero$ when $h(\bx)=0$, in Lemma~\ref{lem:BF-invariance} holds. More importantly, it ensures the construction of a Lipschitz continuous controller~$\bk$~\cite{MJ:18}. 

CBFs provide a constructive way to synthesize a safe controller ${\map{\bk}{\real^n}{\real^m}}$ through modifying a desired controller ${\map{\bk_\des}{\real^n}{\real^m}}$ using the quadratic program (QP):
\begin{equation} \label{eq:CBF-QP}
\begin{aligned}
    \bk(\bx)= \argmin_{\bu\in\real^m} & \quad \tfrac{1}{2} \|\bu-\bk_\des(\bx)\|_{\bGamma}^2 \\
    \text{s.t.} & \quad L_\bf h(\bx) +L_\bg h(\bx) \bu \geq -\alpha(h(\bx)),
\end{aligned}
\end{equation}
with a symmetric positive definite weight matrix ${\bGamma \in \real^{m 
\times m}}$. This CBF-QP can be solved explicitly in the form~\cite{CohenLCSS23, TGM-SKK-JC-KD-KLH-ADA:25}:
\begin{equation}
\begin{aligned} \label{eq:safetyfilter}
    \bk(\bx) & = \bk_\des(\bx) + \lambda \big( a(\bx), \|\bb(\bx)\|_{\bGamma} \big) \bb(\bx), \\
    a(\bx) & = L_{\bf}h(\bx) + L_{\bg}h(\bx) \bk_\des(\bx) + \alpha \big( h(\bx) \big), \\
    \bb(\bx) & = \bGamma^{-1} L_{\bg}h(\bx)^\top, \\
    \lambda(a,& b) =
    \begin{cases}
        0 & \mathrm{if}\ b = 0, \\ 
        \frac{\relu(-a)}{b^2} & \mathrm{otherwise}.
    \end{cases}
\end{aligned}
\end{equation}
Importantly, this controller is locally Lipschitz on a neighborhood of $\Cc$~\cite{MJ:18} because of the strictness in definition~\eqref{eq:CBC}.
\begin{corollary}\longthmtitle{Safety}
    If $h$ is a CBF for~\eqref{sys:ctrl_affine} on $\Cc$, then $\Cc$ is a control invariant set. Furthermore, it is a safe set if $\Cc\subseteq\Sc$.
\end{corollary}

To construct a function $h$ meeting the CBF condition~\eqref{eq:CBC}, it is often easier to verify that:
\begin{equation}\label{eq:CBC_check}
    L_\bg h(\bx) = \bzero \implies L_\bf h(\bx)>-\alpha(h(\bx)),
\end{equation} 
which is equivalent to~\eqref{eq:CBC}.
This condition states that $h$ needs only be analyzed where $L_\bg h(\bx) =\bzero$, reducing the set of states where one must verify the validity of a candidate CBF.

\section{Optimal-Decay CBFs} \label{sec:OD-CBF}

Verifying \eqref{eq:CBC_check} can be challenging as, even if $h$ is a CBF, finding an ${\alpha\in\Ke}$ satisfying \eqref{eq:CBC_check} is often not straightforward. In certain cases, $\alpha$ can be scaled (e.g., multiplied by a constant) to help satisfy the inequality, but because $\alpha$ is a function of $h$ rather than the state $\bx$, the correct $\alpha$ may not be found by scaling alone. To this end, a new development in the CBF framework, often used in practice to address this issue, is the \textit{optimal-decay} CBF-QP~\cite{ZengACC21}.

Building on the CBF-QP control formulation, the optimal-decay CBF-QP simultaneously solves for a scaling factor multiplying $\alpha$, denoted by $\theta$, along with the controller $\bk$: 
\begin{equation} \label{eq:OD-CBF-QP}
\begin{aligned}
    \begin{bmatrix}
        \bk(\bx) \\ \theta(\bx)
    \end{bmatrix}
    & = \argmin_{\substack{\bu\in\real^{m} \\ \omega\in\realnonneg}}
    \tfrac{1}{2}
    \|\bu-\bk_\des(\bx)\|_{\bGamma}^2 +
    \tfrac{1}{2}
    p (\omega-\theta_\des)^2 \\
    &\qquad \text{s.t.}\quad L_\bf h(\bx) +L_\bg h(\bx) \bu \geq -\omega\alpha(h(\bx)), \\
    &\qquad \qquad ~ \omega \geq \theta_\des,
\end{aligned}
\end{equation}
where ${\theta_\des \geq 0}$ is the desired value of $\theta(\bx)$, whereas the symmetric positive definite matrix ${\bGamma \in \real^{m \times m}}$ and ${p>0}$ are weights.
In effect, this control formulation enforces:
\begin{equation}\label{eq:OD-CBF_enforces}
L_\bf h(\bx) +L_\bg h(\bx) \bk(\bx) \geq -\theta(\bx)\alpha(h(\bx)),
\end{equation}
for a decay rate ${\omega = \theta(\bx)}$ given by ${\map{\theta}{\R^n}{\realnonneg}}$,  
which still guarantees the forward invariance of $\Cc$ as will be shown shortly. Although the optimal-decay CBF-QP is a common technique used in practice~\cite{ZengACC21}, formal results regarding its properties are relatively sparse.
Here, we fill this gap by providing a rigorous analysis of the control formulation. 

\subsection{Optimal-Decay BFs}

We begin with the formalization of optimal-decay BFs.
\begin{definition}\longthmtitle{Optimal-Decay BF}
The function $h$ from~\eqref{eq:safeset} is an \textbf{optimal-decay barrier function} (OD-BF) for~\eqref{sys:closedloop} on~$\Cc$ if there exists ${\alpha \in \Kc}$ and a function $\map{\theta}{\real^n}{\realnonneg}$ such that~\eqref{eq:OD-CBF_enforces} holds
for all $\bx\in\mathcal{C}$.
\end{definition}

\begin{theorem}\label{thm:converse_OD-BF} Let the system dynamics $\bF$ for \eqref{sys:closedloop} be locally Lipschitz and suppose that ${\frac{\partial h}{\partial\bx}(\bx)\neq \bzero}$ when ${h(\bx)=0}$. Then, $\Cc$ is forward invariant if and only if $h$ is an OD-BF.
\end{theorem}
\begin{proof}
    We will establish that~\eqref{eq:OD-CBF_enforces} is equivalent to~\eqref{eq:nagumo}. If~\eqref{eq:OD-CBF_enforces} holds, then we can directly verify~\eqref{eq:nagumo} with the direct substitution of $h(\bx)=0$. On the other hand, if~\eqref{eq:nagumo} holds, we may pick any class-$\Kc$ function $\alpha$ and define:
    $$
    \theta(\bx) = \begin{cases} 
    \frac{\relu(-L_\bF h(\bx))}{\alpha(h(\bx))} & h(\bx)\neq 0, \\
    0 & h(\bx)=0.
    \end{cases}
    $$
    This function is nonnegative for all ${\bx \in \Cc}$. In addition, when $h(\bx)=0$, \eqref{eq:nagumo} implies~\eqref{eq:OD-CBF_enforces}. When $h(\bx)\neq 0$, we have:
    $$
    L_\bF h(\bx) \geq -\relu(-L_\bF h(\bx)) = -\theta(\bx)\alpha(h(\bx)) 
    $$
    as desired. 
    With the equivalence between \eqref{eq:OD-CBF_enforces} and \eqref{eq:nagumo}, the result in the theorem is then precisely as in Lemma~\ref{lem:not-nagumo}.
\end{proof}
Theorem~\ref{thm:converse_OD-BF} shows that OD-BFs have a stronger connection to set forward invariance than standard BFs do. Without the the need to impose additional assumption about superlevel set compactness, cf. Proposition~\ref{prop:converse_BF}, OD-BFs close the gap between the BF framework and Lemma~\ref{lem:not-nagumo}. This foundation adds more flexibility in rendering a set forward invariant, which we will demonstrate in the context of control systems. 

\subsection{Optimal-Decay CBFs}
Building on the development of OD-BFs, we formally define optimal decay CBFs as follows.

\begin{definition}\longthmtitle{Optimal-Decay CBF}
The function $h$ from \eqref{eq:safeset} is an \textbf{optimal-decay control barrier function} (OD-CBF) for~\eqref{sys:ctrl_affine} on $\Cc$ if there exists ${\alpha \in \Ke}$ such that for all ${\bx \in \Cc}$:
\begin{equation}\label{eq:OD-CBC}
\sup_{\substack{\bu\in\real^{m},~\omega\in\realnonneg}} L_{\bf}h(\bx) +L_{\bg}h(\bx) \bu +\omega\alpha(h(\bx)) > 0.
\end{equation}
\end{definition}

The first result we establish is analogous to the condition in \eqref{eq:CBC_check}, which provides an alternative, but equivalent, characterization of OD-CBFs.
\begin{lemma}\label{lem:OD-CBC_check}
    The function $h$ from \eqref{eq:safeset} is an optimal-decay CBF for~\eqref{sys:ctrl_affine} on $\Cc$ if and only if:
    \begin{equation}\label{eq:OD-CBC_check}
        h(\bx) = 0\ \wedge\ L_{\bg}h(\bx)=\bzero \implies L_{\bf}h(\bx) > 0.
    \end{equation}
\end{lemma}
\begin{proof}
    We define the shorthand notation:
    \begin{equation*}
        \rho(\bx,\bu,\omega) \coloneqq L_{\bf}h(\bx) + L_{\bg}h(\bx)\bu + \omega\alpha(h(\bx))
    \end{equation*}
    and then note that for all $\bx\in\mathcal{C}$:
    \begin{equation*} 
        \sup_{\substack{\bu\in\real^{m} \\ \omega\in\realnonneg}}\rho(\bx,\bu,\omega) = \begin{cases}
            L_{\bf}h(\bx) & \mathrm{if}\ h(\bx) = 0 \ \wedge \\
            & \quad L_{\bg}h(\bx)=\bzero, \\
            \infty & \mathrm{otherwise}.
        \end{cases}
    \end{equation*}
    Using this equation, we observe that
    \eqref{eq:OD-CBC_check} is equivalent to:
    $ 
        \sup_{\substack{\bu\in\real^{m}, \omega\in\realnonneg}}\rho(\bx,\bu,\omega) > 0,
    $
    which is precisely~\eqref{eq:OD-CBC}. 
\end{proof}

A main advantage of optimal-decay CBFs is that they can be verified by checking the condition~\eqref{eq:OD-CBC_check} which is equivalent to \eqref{eq:OD-CBC} (when there are no input bounds, i.e., ${\bu \in \real^m}$). 
This verification does not require searching for $\alpha$ and reduces the set of points where the validity of the CBF must be checked.
In addition, once the condition is verified, the resulting controller $\bk$ given in~\eqref{eq:OD-CBF-QP} is locally Lipschitz and admits a closed-form solution, which we establish next.
\begin{lemma} \label{lem:QPsolution}
    If $h$ is an OD-CBF for \eqref{sys:ctrl_affine} on $\Cc$,
    then the QP~\eqref{eq:OD-CBF-QP} is feasible for all ${\bx \in \Cc}$.
    The controller $\bk$ generated by~\eqref{eq:OD-CBF-QP} by can be expressed as:
    \begin{equation} \label{eq:OD-safetyfilter}
        \begin{aligned}
            \bk(\bx) & = \bk_\des(\bx) + \phi\big(a(\bx), \|\bb(\bx)\|_{\bGamma}, c(\bx) \big) \bb(\bx), \\ 
            a(\bx) & = L_{\bf}h(\bx) + L_{\bg}h(\bx)\bk_\des(\bx) + \theta_\des\alpha(h(\bx)), \\
            \bb(\bx) & = \bGamma^{-1} L_{\bg}h(\bx)^\top, \quad
            c(\bx) = p^{-1} \alpha(h(\bx)), \\ 
            \phi(a,& b,c) = \begin{cases}
                0 & \mathrm{if}\ b=0 \ \wedge \ c \leq 0, \\ 
                \frac{\relu(-a)}{b^2 + p \relu(c)^2} & \mathrm{otherwise},
            \end{cases}
        \end{aligned}
    \end{equation}
    whereas the corresponding optimal decay rate is:
    \begin{equation} \label{eq:OD-decayrate}
    \begin{aligned}
        \theta(\bx) & = \theta_\des + \psi \big( a(\bx), \|\bb(\bx)\|_{\bGamma}, c(\bx) \big), \\
        \chi(a,& b,c) =
        \begin{cases}
            0 & \mathrm{if}\ b=0 \ \wedge \ c \leq 0, \\ 
            \frac{\relu(-a) \relu(c)}{b^2 + p c^2} & \mathrm{otherwise}.
        \end{cases}
    \end{aligned}
    \end{equation}
\end{lemma}

\begin{proof}
    First, we notice that the first constraint in~\eqref{eq:OD-CBF-QP} can be represented using $a$, $\bb$, and $c$ in~\eqref{eq:OD-safetyfilter} as:
    \begin{equation*}
    \begin{aligned}
        & L_\bf h(\bx) + L_\bg h(\bx) \bu + \omega \alpha(h(\bx)) \\
        & \qquad = a(\bx) + \bb(\bx)^\top \bGamma \big( \bu - \bk_\des(\bx) \big) + c(\bx) p (\omega - \theta_\des).
    \end{aligned}
    \end{equation*}
    The QP~\eqref{eq:OD-CBF-QP} is feasible if and only if:
    \begin{equation} \label{eq:feasbility}
        \bb(\bx) = \bzero \; \wedge \; c(\bx) \leq 0 \implies a(\bx) \geq 0.
    \end{equation}
    Because the OD-CBF property~\eqref{eq:OD-CBC_check} is equivalent to ${\bb(\bx) = \bzero \; \wedge \; c(\bx) = 0 \implies a(\bx) > 0}$, feasibility holds for all ${\bx \in \Cc}$, i.e., when ${c(\bx) \geq 0}$.
    
    We solve the QP~\eqref{eq:OD-CBF-QP} using the Karush-Kuhn-Tucker (KKT) conditions from~\cite{SB-LV:09}:
    \begin{equation} \label{eq:KKT}
    \begin{aligned}
        a(\bx) + \bb(\bx)^\top \hat{\bk}(\bx) + c(\bx) \hat{\theta}(\bx) & \geq 0, \\
        p^{-1} \hat{\theta}(\bx) \geq 0,~
        \lambda_{\bu}(\bx) \geq 0,~ 
        \lambda_{\omega}(\bx) & \geq 0, \\
        \lambda_{\bu}(\bx) \big( a(\bx) + \bb(\bx)^\top \hat{\bk}(\bx) + c(\bx) \hat{\theta}(\bx) \big) & = 0, \\
        \lambda_{\omega}(\bx) p^{-1} \hat{\theta}(\bx) & = 0, \\
        \hat{\bk}(\bx) - \lambda_{\bu}(\bx) \bGamma \bb(\bx) & = \bzero, \\
        \hat{\theta}(\bx) - \lambda_{\bu}(\bx) p c(\bx) - \lambda_{\omega}(\bx) & = 0,
    \end{aligned}
    \end{equation}
    where:
    \begin{equation} \label{eq:OD-transformation}
        \hat{\bk}(\bx) = \bGamma \big( \bk(\bx) - \bk_\des(\bx) \big), \quad
        \hat{\theta}(\bx) = p (\theta(\bx) - \theta_\des).
    \end{equation}
    
    Here ${\map{\lambda_{\bu}}{\real^n}{\real}}$ and ${\map{\lambda_{\omega}}{\real^n}{\real}}$ are Lagrange multipliers.
    Based on their values, we distinguish four cases.
    Case 1: if ${\lambda_{\bu}(\bx) = 0}$, ${\lambda_{\omega}(\bx) = 0}$, the KKT conditions~\eqref{eq:KKT} give:
    \begin{equation*}
        \hat{\bk}(\bx) = \bzero, \quad
        \hat{\theta}(\bx) = 0, \quad
        a(\bx) \geq 0.
    \end{equation*}
    Case 2: if ${\lambda_{\bu}(\bx) > 0}$, ${\lambda_{\omega}(\bx) = 0}$, we obtain:
    \begin{equation*}
    \begin{aligned}
        \hat{\bk}(\bx) & = \frac{-a(\bx)}{\|\bb(\bx)\|_{\bGamma}^2 + p c(\bx)^2} \bGamma \bb(\bx), \quad
        & a(\bx) < 0, \\
        \hat{\theta}(\bx) & = \frac{-a(\bx)}{\|\bb(\bx)\|_{\bGamma}^2 + p c(\bx)^2} p c(\bx), \quad
        & c(\bx) \geq 0.
    \end{aligned}
    \end{equation*}
    Case 3: if ${\lambda_{\bu}(\bx) = 0}$, ${\lambda_{\omega}(\bx) > 0}$, the sixth and eighth lines of the KKT conditions~\eqref{eq:KKT} contradict, and this case cannot hold.
    Case 4: if ${\lambda_{\bu}(\bx) > 0}$, ${\lambda_{\omega}(\bx) > 0}$, then we get:
    \begin{equation*}
    \begin{aligned}
        \hat{\bk}(\bx) & = \frac{-a(\bx)}{\|\bb(\bx)\|_{\bGamma}^2} \bGamma \bb(\bx), \quad
        & a(\bx) < 0, \\
        \hat{\theta}(\bx) & = 0, \quad
        & c(\bx) < 0.
    \end{aligned}
    \end{equation*}

    Overall, these cases can be combined into:
    \begin{equation*}
        \hat{\bk}(\bx) =
        \begin{cases}
            \bzero & \mathrm{if}\ a(\bx) \geq 0, \\
            \frac{-a(\bx)}{\|\bb(\bx)\|_{\bGamma}^2 + p c(\bx)^2} \bGamma \bb(\bx) & \mathrm{if}\ a(\bx) < 0 \; \wedge \; c(\bx) \geq 0, \\
            \frac{-a(\bx)}{\|\bb(\bx)\|_{\bGamma}^2} \bGamma \bb(\bx) & \mathrm{if}\ a(\bx) < 0 \; \wedge \; c(\bx) < 0,
        \end{cases}
    \end{equation*}
    \begin{equation*}
        \hat{\theta}(\bx) =
        \begin{cases}
            0 & \mathrm{if}\ a(\bx) \geq 0, \\
            \frac{-a(\bx)}{\|\bb(\bx)\|_{\bGamma}^2 + p c(\bx)^2} p c(\bx) & \mathrm{if}\ a(\bx) < 0 \; \wedge \; c(\bx) \geq 0, \\
            0 & \mathrm{if}\ a(\bx) < 0 \; \wedge \; c(\bx) < 0.
        \end{cases}
    \end{equation*}
    Given the feasibility property~\eqref{eq:feasbility}, these can be written equivalently using $\phi$ in~\eqref{eq:OD-safetyfilter} and $\chi$ in~\eqref{eq:OD-decayrate} as:
    \begin{equation} \label{eq:OD-CBF-QP_transformed_solution}
    \begin{aligned}
        \hat{\bk}(\bx) & = \phi \big( a(\bx), \|\bb(\bx)\|_{\bGamma}, c(\bx) \big) \bGamma \bb(\bx), \\
        \hat{\theta}(\bx) & = \chi \big( a(\bx), \|\bb(\bx)\|_{\bGamma}, c(\bx) \big) p.
    \end{aligned}
    \end{equation}
    Substituting this into~\eqref{eq:OD-transformation} yields~\eqref{eq:OD-safetyfilter}-\eqref{eq:OD-decayrate}.
\end{proof}
\begin{remark}\label{rmk:denominator}
    Note the difference between the OD-CBF-QP controller~\eqref{eq:OD-safetyfilter} and the standard CBF-QP~\eqref{eq:safetyfilter}.
    The denominator ${\|\bb(\bx)\|_{\bGamma}^2}$ in~\eqref{eq:safetyfilter} is zero if and only if ${L_\bg h(\bx) = \bzero}$, whereas the denominator ${\|\bb(\bx)\|_{\bGamma}^2 + p \relu(c(\bx))^2}$ in~\eqref{eq:OD-safetyfilter} is zero for some ${\bx \in \Cc}$ if and only if ${L_\bg h(\bx) = \bzero}$ and ${h(\bx) = 0}$, which occurs for a smaller set of states. 
    Practically speaking, this leads to better numerical conditioning of the OD-CBF-QP, as demonstrated later in our numerical results.
\end{remark}
\begin{proposition}\label{prop:lipschitz}
  If $h$ is an OD-CBF for \eqref{sys:ctrl_affine} on $\Cc$,
  and if $\alpha$, $\frac{\partial h}{\partial \bx}$ and $\bk_\des$ are locally Lipschitz, then the controller in~\eqref{eq:OD-safetyfilter} is locally Lipschitz at each $\bx\in\Cc$.
\end{proposition}
\begin{proof}
    First, we show that the optimal-decay CBF-QP is defined locally for each $\bx$ on the boundary of the set $\Cc$ where ${h(\bx)=0}$.
    If ${L_\bg h(\bx) \neq \bzero}$, then there exists an open ball centered at $\bx$ where ${L_\bg h(\bx) \neq \bzero}$ due to the continuity of $L_\bg h$.
    Therefore, condition~\eqref{eq:OD-CBC} holds locally.
    If ${L_\bg h(\bx) = \bzero}$, then condition~\eqref{eq:OD-CBC_check} suggests that ${L_\bf h(\bx) > 0}$ is satisfied, which also holds locally in an open ball around $\bx$ due to the continuity of $L_\bf h$.
    Thus, again, condition~\eqref{eq:OD-CBC} holds locally.
    Hence, the optimal-decay CBF-QP is defined locally in an open ball around any ${\bx \in \Cc}$.
    
    Considering the controller~\eqref{eq:OD-safetyfilter}, all the functions involved in its definition (i.e., $a$, $\bb$, $c$, and $\relu$) are locally Lipschitz.
    Thus, $\bk$ is locally Lipschitz where its denominator is nonzero, i.e., for all $\bx$ where 
       $ D(\bx) \triangleq \|\bb(\bx)\|_{\bGamma}^2 + p \relu(c(\bx))^2\neq 0.$  
    On the other hand, ${D(\bx) = 0}$ if and only if $L_\bg h(\bx) =\bzero$ and $h(\bx)\leq 0$. For $\bx\in\Cc$, this is true only when $L_\bg h(\bx) =\bzero$ and ${h(\bx) = 0}$, implying $L_\bf h(\bx)>0$ based on \eqref{eq:OD-CBC_check}. Therefore, 
    $
    -a(\bx) =
    -L_\bf h(\bx)<0,
    $
    for any ${\bx \in \Cc}$ satisfying ${D(\bx) = 0}$.
    Due to the continuity of $a$, ${-a(\bx)<0}$ also holds locally in an open ball around such state $\bx$.
    Thus, ${\relu(-a(\bx))=0}$ in~\eqref{eq:OD-safetyfilter}, making the function $\bk$ locally Lipschitz and concluding the proof. 
\end{proof}

Now, we show that any controller satisfying~\eqref{eq:OD-CBF_enforces} for a ${\theta(\bx) \geq 0}$ 
renders $\Cc$ forward invariant.

\begin{theorem} \label{thm:OD-CBF}
If $h$ is an OD-CBF for~\eqref{sys:ctrl_affine} on $\Cc$, then any locally Lipschitz controller ${\map{\bk}{\real^n}{\real^m}}$ that satisfies~\eqref{eq:OD-CBF_enforces} with some ${\map{\theta}{\real^n}{\realnonneg}}$,
for all ${\bx \in \Cc}$, renders the set $\Cc$ forward invariant for~\eqref{sys:closedloop}.
\end{theorem} 

\begin{proof}
    Based on~\eqref{eq:OD-CBC}, when ${h(\bx)=0}$, we have 
        $\sup_{\bu\in\real^{m}} L_{\bf}h(\bx) +L_{\bg}h(\bx) \bu > 0$,
    which implies that $\frac{\partial h}{\partial \bx}(\bx)\neq \bzero$. 
    From~\eqref{eq:OD-CBF_enforces}, we may deduce that: 
    $$
    h(\bx) = 0 \implies L_\bf h(\bx) + L_\bg h(\bx) \bk(\bx) \geq 0,
    $$
    which implies the forward invariance of $\Cc$ by Lemma \ref{lem:not-nagumo}.
\end{proof}
Combining the previous results allows for establishing that $\Cc$ is control invariant.
\begin{corollary}
    If $h$ is an OD-CBF for~\eqref{sys:ctrl_affine} on $\Cc$, then $\Cc$ is a control invariant set. Furthermore,
    it is a safe set if $\Cc\subseteq\Sc$.
\end{corollary}
\begin{proof}
    Since~\eqref{eq:OD-CBC} holds for all ${\bx\in\Cc}$, the QP~\eqref{eq:OD-CBF-QP} is feasible by Lemma~\ref{lem:QPsolution} and it enforces \eqref{eq:OD-CBF_enforces} for all ${\bx\in\Cc}$ (considering any ${\theta_\des \geq 0}$ and any locally Lipschitz $\bk_d$). Moreover, this QP-based controller is locally Lipschitz continuous by Proposition \ref{prop:lipschitz}, and
    renders the set $\Cc$ forward invariant by Theorem \ref{thm:OD-CBF}.
    Thus, the state-feedback controller~\eqref{eq:OD-CBF-QP} produces a control signal $\bu(t)=\bk(\bx(t))$ that renders $\Cc$ forward invariant, and we conclude that $\Cc$ is a control invariant set.
    By definition, if $\Cc\subseteq\Sc$, then it is also a safe set.
\end{proof}

Moreover, apart from forward and control invariance, the stability of the set $\Cc$ can also be established for the system~\eqref{sys:closedloop} with the controller~\eqref{eq:OD-CBF-QP} under additional assumptions.
\begin{proposition}
    Let $h$ be an OD-CBF for \eqref{sys:ctrl_affine} on $\Cc$ and suppose there exists a set $\mathcal{E}\supset\mathcal{C}$ such that $L_{\bg}h(\bx)\neq\bzero$ for all $\bx\in\mathcal{E}\setminus\mathcal{C}$. Then, the QP \eqref{eq:OD-CBF-QP} is feasible for all $\bx\in\mathcal{E}$ with any $\alpha\in\Ke$ and, if $\Cc$ is compact and $\theta_{\rm{d}}>0$, the resulting controller renders $\Cc$ asymptotically stable for~\eqref{sys:closedloop}.
\end{proposition}
\begin{proof}
    Since $h$ is an OD-CBF, \eqref{eq:OD-CBF-QP} is feasible for all $\bx\in\mathcal{C}$. Further, since $L_{\bg}h(\bx)\neq\bzero$ for all $\bx\in\mathcal{E}\setminus\mathcal{C}$, the OD-CBF condition \eqref{eq:OD-CBC} holds on $\bx\in\mathcal{E}\setminus\mathcal{C}$, implying feasibility of \eqref{eq:OD-CBF-QP} for all  $\bx\in\mathcal{E}\setminus\mathcal{C}$. Together, this implies that \eqref{eq:OD-CBF-QP} is feasible for all $\bx\in\mathcal{E}$, guarantees that the resulting controller $\bk\,:\,\mathcal{E}\rightarrow\mathbb{R}^m$ satisfies \eqref{eq:OD-CBF_enforces} for all $\bx\in\mathcal{E}$, and renders $\Cc$ is forward invariant (Theorem \ref{thm:OD-CBF}). To show asymptotic stability, consider the Lyapunov function candidate:
    \begin{equation*}
        V(\bx) \coloneqq
        \begin{cases}
            0 & \text{if } \bx\in\mathcal{C} \\ 
            h(\bx)^2 & \text{if } \bx\in\mathcal{E}\setminus\mathcal{C},
        \end{cases}
    \end{equation*}
    which is continuously differentiable and satisfies ${V(\bx)=0\iff\bx\in\Cc}$ and ${V(\bx)>0}$ for ${\bx\in\mathcal{E}\setminus\mathcal{C}}$ so that $V$ is positive definite with respect to $\Cc$. The derivative of $V$ along the trajectories of the closed-loop system~\eqref{sys:closedloop} is:
    \begin{equation*}
        \dot{V}(\bx) =
        \begin{cases}
            0 & \text{if } \bx\in\mathcal{C}  \\ 
            2h(\bx)\big(L_{\bf}h(\bx) + L_{\bg}h(\bx)\bk(\bx) \big) & \text{if } \bx\in\mathcal{E}\setminus\mathcal{C}.
        \end{cases}
    \end{equation*}
    It then follows from \eqref{eq:OD-CBF-QP} that when $\bx\in\mathcal{E}\setminus\mathcal{C}$ we have:
    \begin{equation*}
        L_{\bf}h(\bx) + L_{\bg}h(\bx)\bk(\bx) \geq \underbrace{- \theta(\bx)}_{<0}\underbrace{\alpha (h(\bx))}_{<0} > 0,
    \end{equation*}
    which implies that when $\bx\in\mathcal{E}\setminus\mathcal{C}$ we also have:
    \begin{equation*}
        \underbrace{2h(\bx)}_{<0}\underbrace{\big(L_{\bf}h(\bx) + L_{\bg}h(\bx)\bk(\bx) \big)}_{>0} < 0.
    \end{equation*}
    Thus, $\dot{V}(\bx)=0\iff\bx\in\mathcal{C}$ and $\dot{V}(\bx) < 0$ for all $\bx\in\mathcal{E}\setminus\mathcal{C}$, showing that $\dot{V}$ is negative definite with respect to $\mathcal{C}$. Since $\Cc$ is a compact forward invariant set, $V$ is positive definite with respect to $\Cc$, and $\dot{V}$ is negative definite with respect to $\Cc$, it follows from Lyapunov's Theorem for set stability \cite[Corollary 4.7]{WMH-VC:08} that $\Cc$ is asymptotically stable.
\end{proof}

\section{Optimal-Decay CBF Variants}

Different variants of CBFs can also benefit from the proposed OD-CBF framework. Particularly, we will investigate how these results apply to high-order CBFs (HOCBFs)~\cite{QN-KS:16, WeiTAC22, TanTAC22} and rectified CBFs (ReCBFs)~\cite{PO-MHC-TGM-ADA:25-csl}. These methods are constructive approaches for generating CBF candidates that address safety constraints with higher relative degrees.

Given a safety constraint $\Sc$ in~\eqref{eq:safety_constraint}, the construction of CBFs typically involve the safety constraint function~$\psi$. For example, the first natural CBF candidate is the safety constraint function~$\psi$ itself. Under the new OD-CBF framework, we only need to check if~\eqref{eq:OD-CBC_check} holds. If $\psi$ is an OD-CBF, then $\Sc$ itself is a safe set for~\eqref{sys:ctrl_affine}. 

Nevertheless, an arbitrary constraint function $\psi$ may not be an OD-CBF, particularly when $\psi$ has higher relative degrees.
\begin{definition}\longthmtitle{Relative Degree}
A smooth function ${\map{\psi}{\real^n}{\real}}$ has relative degree $r\in\naturals$ for~\eqref{sys:ctrl_affine} if:
\begin{enumerate}
    \item $L_\bg L_\bf^i\psi(\bx) = \bzero$ for all $i \in \{1,\dots,r-2\}$ and ${\bx \in \real^n}$;
    \item $L_\bg L_\bf^{r-1}\psi(\bx)\neq \bzero$ for some $\bx\in\real^n$.
\end{enumerate}    
\end{definition}
If $\psi$ has relative degree ${r \geq 2}$, then ${L_\bg \psi(\bx) \equiv \bzero}$, and the OD-CBF condition~\eqref{eq:OD-CBC_check} for $\psi$ reduces to
${\psi(\bx) = 0 \implies L_\bf \psi(\bx) >0}$,
implying that $\Sc$ must already be forward invariant without any control. When $\psi$ has higher relative degrees, we cannot rely on $\psi$ being an OD-CBF, but we may use other methods to construct an OD-CBF based on $\psi$. For simplicity, we will limit ourselves to relative degree $r=2$ safety constraints in this paper.

\subsection{Optimal-Decay HOCBF}

First, we demonstrate the power of the optimal-decay CBF framework in the context of HOCBFs~\cite{QN-KS:16, WeiTAC22, TanTAC22}. 
For relative degree ${r=2}$, instead of directly using $\psi$, the HOCBF framework defines the following candidate CBF:
\begin{equation} \label{eq:HOCBF}
    h(\bx) = L_\bf \psi(\bx) + \alpha_1(\psi(\bx)),
\end{equation}
with some continuously differentiable class-$\Ke$ function $\alpha_1$. 

\begin{definition}\longthmtitle{Optimal-Decay HOCBF}
    The function $\psi$ from~\eqref{eq:safety_constraint} is an optimal-decay high-order control barrier function (OD-HOCBF) of degree ${r=2}$ for~\eqref{sys:ctrl_affine} on $\Sc \cap \Cc$ if it has relative degree $r=2$ and there exist $\alpha_1,\alpha\in\Ke$
    such that $h$ in \eqref{eq:HOCBF} satisfies for all $\bx\in\Sc\cap\Cc$ that:
    \begin{equation}\label{eq:OD-HOCBC}
        \sup_{\substack{\bu\in\real^{m},~\omega\in\realnonneg}} L_{\bf}h(\bx) +L_{\bg}h(\bx) \bu +\omega\alpha(h(\bx)) > 0.
    \end{equation}
\end{definition}
The OD-HOCBF condition~\eqref{eq:OD-HOCBC} for $\psi$ is similar to the OD-CBF condition~\eqref{eq:OD-CBC} for $h$. The key difference is that it only needs to hold on $\Sc \cap \Cc$, rather than the entire $\Cc$.
Next, we state a sufficient condition for $\psi$ to be an OD-HOCBF.

\begin{corollary}\label{coro:OD-HOCBF_simple_check}
Consider a constraint function~$\psi$ from~\eqref{eq:safety_constraint} with relative degree ${r=2}$ for~\eqref{sys:ctrl_affine}. If ${L_\bg L_\bf\psi(\bx)\neq \bzero}$ for each ${\bx\in\Sc\cap\Cc}$ satisfying ${h(\bx)=0}$, then $\psi$ is an OD-HOCBF of degree ${r=2}$ for~\eqref{sys:ctrl_affine} on $\Sc\cap\Cc$.
\end{corollary}
\begin{proof}
    Because ${L_\bg h(\bx)  = L_\bg L_\bf\psi(\bx)}$,
    \eqref{eq:OD-HOCBC} holds for any ${\alpha \in \Ke}$ under the assumption of the corollary, and thus, $\psi$ is an OD-HOCBF.
\end{proof}

With an OD-HOCBF, we establish that ${\Sc \cap \Cc}$ is a safe set.
\begin{proposition}\label{prop:OD-HOCBF}
    If $\psi$ is an OD-HOCBF of degree ${r=2}$ for~\eqref{sys:ctrl_affine} on $\Sc\cap\Cc$ and $\frac{\partial\psi}{\partial \bx}(\bx)\neq \bzero$ when $\psi(\bx)=h(\bx)=0$, then $\Sc\cap\Cc$ is a safe set.
\end{proposition}
\begin{proof}
    Let the QP-based controller $\bk$ in 
    \eqref{eq:OD-CBF-QP} define the closed-loop vector field~\eqref{sys:closedloop}. Since $h$ is an OD-CBF, Proposition~\ref{prop:lipschitz} guarantees that $\bk$ is locally Lipschitz at each ${\bx\in\Sc\cap\Cc}$, and so is the closed-loop vector field~$\bF$. We first note that in the interior of the set ${\Sc\cap\Cc}$,
    the vector field is trivially in the set's tangent cone~\cite[Def. 4.9]{Blanchini}.
    
    Now, consider the points on the boundary of the set
    ${\Sc\cap\Cc}$.
    If ${\psi(\bx) = 0}$ and ${h(\bx) > 0}$, then by definition~\eqref{eq:HOCBF} of $h$ and noting that ${L_\bg\psi(\bx) = 0}$, we have
    $
    \frac{\partial \psi}{\partial \bx}(\bx) \bF(\bx)=\frac{\partial \psi}{\partial \bx}(\bx) \bf(\bx) > 0
    $
    and
    ${\frac{\partial \psi}{\partial \bx}(\bx) \neq \bzero}$.
    Therefore, the closed-loop vector field $\bF(\bx)$ belongs to the tangent cone at $\bx$, given by:
    $$
    \Tc(\bx) = \setdefB{\bv\in\real^n}{\frac{\partial \psi}{\partial \bx}(\bx) \bv \geq 0}.
    $$
    
    If ${\psi(\bx)>0}$ and ${h(\bx)=0}$, then 
    ${
    \frac{\partial h}{\partial \bx}(\bx) \bF(\bx) \geq 0
    }$
    because the resulting closed-loop system enforces \eqref{eq:OD-CBF_enforces}. Since $h$ is an OD-CBF, we have $\frac{\partial h}{\partial \bx}(\bx)\neq \bzero$ when $h(\bx)=0$ because the left hand side of~\eqref{eq:OD-CBC} would evaluate to zero otherwise. Thus, the closed-loop vector field~$\bF(\bx)$ belongs to the tangent cone:
    $$
    \Tc(\bx) = \setdefB{\bv\in\real^n}{\frac{\partial h}{\partial \bx}(\bx)\bv \geq 0}.
    $$

    Finally, if ${\psi(\bx)=h(\bx)=0}$, the constraint~\eqref{eq:OD-CBF_enforces} enforces that
    ${
    \frac{\partial h}{\partial \bx}(\bx) \bF(\bx) \geq 0
    }$,
    whereas the definition~\eqref{eq:HOCBF} of $h$ gives
    ${
    \frac{\partial \psi}{\partial \bx}(\bx) \bF(\bx)=\frac{\partial \psi}{\partial \bx}(\bx) \bf(\bx) = 0
    }$.
    The assumption that $\frac{\partial\psi}{\partial \bx}(\bx)\neq \bzero$ ensures that the tangent cone is given by:
    $$
    \Tc(\bx) = \setdefB{\bv\in\real^n}{\frac{\partial \psi}{\partial \bx}(\bx)\bv \geq 0~\wedge~\frac{\partial h}{\partial \bx}(\bx)\bv \geq 0},
    $$
    and the vector field $\bF(\bx)$ belongs to it. 

    As a result of Nagumo's theorem~\cite{MN:42}, the set $\Sc\cap\Cc$ is forward invariant under the controller ${\bu =\bk(\bx)}$ in~\eqref{eq:OD-CBF-QP}, because the locally Lipschitz vector field $\bF(\bx)$ is in the tangent cone of the set ${\Sc\cap\Cc}$ for all ${\bx\in\Sc\cap\Cc}$.
    Thus, this set is control invariant, and it is also a safe set as ${\Sc\cap\Cc\subseteq \Sc}$.
\end{proof}
Due to the similarity of the OD-HOCBF condition~\eqref{eq:OD-HOCBC} and~\eqref{eq:OD-CBC}, all results regarding the Lipschitz continuity of the QP-based controllers hold, and are omitted in the interest of space.  We illustrate the flexibility OD-HOCBFs in addressing higher-order safety constraints with the following example.

\subsection{Double Integrator Example}

Consider a double integrator with state ${\bx = (x, \dot x) \in \real^2}$ and control input ${u\in\real}$ with a safety constraint:
\begin{equation}\label{sys:double}
    \dot \bx = \begin{bmatrix}
        \dot x &
        u
    \end{bmatrix}^\top,\quad \psi(\bx) = 1 - x^2. 
\end{equation}
Since $\psi$ has relative degree ${r=2}$ with ${L_\bg L_\bf\psi(\bx)= -2x}$, we employ the HOCBF framework and introduce:
\begin{align*}
h(\bx) & = L_\bf\psi(\bx) +\alpha_1(\psi(\bx))
= -2x\dot x +\alpha_1(1-x^2),
\end{align*}
using a continuously differentiable ${\alpha_1 \in \Ke}$. It has been shown in the literature, see~\cite[Ex. 1]{PO-MHC-TGM-ADA:25-csl}, that there exists no $\alpha$ that makes $h$ satisfy the  condition~\eqref{eq:CBC_check} for CBF.

Nevertheless, we can leverage the OD-HOCBF framework. We can verify the assumption of Corollary~\ref{coro:OD-HOCBF_simple_check} via contradiction: when $L_\bg h(\bx)= -2x =0$, we have $h(\bx) = \alpha_1(1) > 0$. Thus, $\psi$ is an OD-HOCBF. In addition, when $\psi(\bx)=0$, we have $\frac{\partial\psi}{\partial\bx}(\bx) = \begin{bmatrix}
    -2x & 0
\end{bmatrix}^\top =  \begin{bmatrix}
    \pm 2 & 0
\end{bmatrix}^\top\neq \bzero$. As such, $\Sc\cap \Cc$ is a safe set, and we can use the QP-based controller $\bk$ in \eqref{eq:OD-safetyfilter} to render the set forward invariant.

Alternatively, we may use ${\theta(\bx) = 4\dot x^2/\alpha_2(\alpha_1(1))+1}$ with any class-$\Kc^e$ $\alpha_2$. This predefined choice of $\theta$ ensures that ${L_\bf h(\bx)=-2\dot x^2 > -\theta(\bx)\alpha_2(\alpha_1(1))}$ when ${L_\bg L_\bf \psi(\bx)=-2x=\bzero}$. We can design a QP-based controller to optimize for $\bu=\bk(\bx)$ with this $\omega=\theta(\bx)$. The details of this technique will be in our future work, as potentially this technique may allow the analysis of safety performance, due to its simpler nature in comparison to the complex $\theta$ obtained from a QP in~\eqref{eq:OD-decayrate}.

We compute QP-based controllers using the standard HOCBF framework and OD-HOCBF framework, as well as the one with predefined $\theta$. The relevant parameters used are: $\bk_\des (\bx) = \bzero$, $\Gamma = \bI$, $p=1$, $\alpha_1(s)=\alpha_2(s) = 2s$, $\theta_\des=1$, and $\varepsilon=0.1$. Fig.~\ref{fig:double_integrator} shows that the standard HOCBF-QP fails as it becomes unbounded and discontinuous where $L_\bg h(\bx)=\bzero$. On the other hand, the controllers from the OD-HOCBF are well-behaved. We also simulate the control signal for the system with an initial condition
$\bx_0=(
    -0.1, 1.5)$.
The results support the statements in Remark~\ref{rmk:denominator}.
Another CBF construction method, the rectified CBF (ReCBF), can also benefit from this observation, which we discuss next.

\begin{figure}
    \centering
    \includegraphics[width=\linewidth]{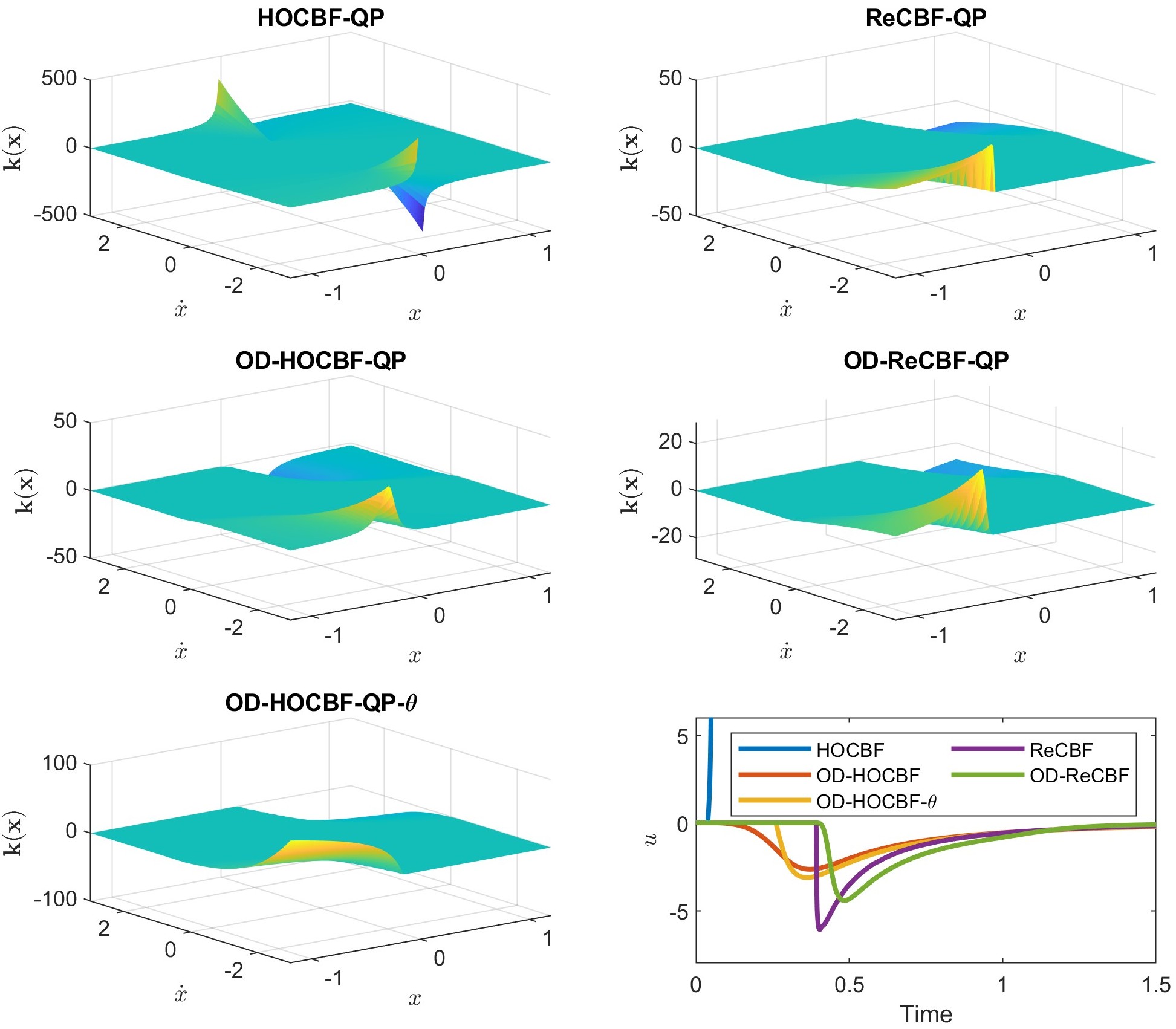}
    \vspace{-0.8cm}
    \caption{Different safe controllers~$\bk$ for the double integrator system~\eqref{sys:double}. }
    \label{fig:double_integrator}
    \vspace{-0.6cm}
\end{figure}

\subsection{Optimal-Decay ReCBF}

For a constraint function $\psi$ with relative degree $r=2$, the ReCBF framework~\cite{PO-MHC-TGM-ADA:25-csl} defines the CBF candidate:
\begin{equation}\label{eq:ReCBF}
    h(\bx) = \psi(\bx) - \relu \Big( - C_1 \big( L_\bf \psi(\bx) + \alpha_1 (\psi(\bx)) -\varepsilon\big)^3 \Big),
\end{equation}
with a class-$\Ke$ $\alpha_1$, a positive constant $C_1>0$, and a well-posedness parameter ${\varepsilon>0}$ introduced to facilitate the strict satisfaction of CBF condition~\eqref{eq:CBC}.
For shorthand notation, we use ${\Gamma(s) = \relu(-C_1s^3)}$, which is twice continuously differentiable with ${\Gamma'(s) = -\relu(-3C_1s\vert s\vert )}$.
It was shown in~\cite{PO-MHC-TGM-ADA:25-csl} that if:
$$
L_\bg L_\bf \psi(\bx) = 0 \implies L_\bf\psi(\bx)\geq -\alpha_1(\psi(\bx))+\varepsilon,
$$
for all ${\bx \in \Cc}$, then $h$ is a CBF.
The following optimal-decay counterpart relaxes this condition to:
\begin{equation}\label{eq:OD-ReCBC}
    h(\bx)=0 \;\wedge\; L_\bg L_\bf\psi(\bx)= 0 \implies L_\bf\psi(\bx)\geq -\alpha_1(\psi(\bx))+\varepsilon. 
\end{equation}

\begin{proposition} \label{prop:OD-ReCBF}
    Consider a constraint function~$\psi$ from~\eqref{eq:safety_constraint} with relative degree ${r=2}$ for~\eqref{sys:ctrl_affine} and   $h$ in~\eqref{eq:ReCBF}.
    If~\eqref{eq:OD-ReCBC} holds for all ${\bx \in \Cc \subset \Sc}$, then $h$ is an OD-CBF for~\eqref{sys:ctrl_affine} on $\Sc$. 
\end{proposition}
\begin{proof}
First, we establish that for $h$ in~\eqref{eq:ReCBF}:
\begin{equation} \label{eq:ReCBF_Lgh}
    L_\bg h(\bx) = -\Gamma'(L_\bf \psi(\bx) + \alpha_1 (\psi(\bx)) -\varepsilon) L_\bg L_\bf\psi(\bx).
\end{equation}
Therefore,
we have
the contrapositive (maintaining the assumption $h(\bx)=0$) of the property~\eqref{eq:OD-ReCBC} as:
\begin{align*}
& h(\bx) = 0 \;\wedge\; L_\bf\psi(\bx) < -\alpha_1(\psi(\bx))+\varepsilon \\
& \qquad \implies L_\bg L_\bf\psi(\bx)\neq \bzero
\implies L_\bg h(\bx)\neq \bzero,
\end{align*}
where the second implication is due to~\eqref{eq:ReCBF_Lgh} and ${\Gamma'(s) \neq 0}$ when ${s<0}$. Consequently,  the following:
$$
h(\bx)=0~\wedge~L_\bg h(\bx)=\bzero \implies L_\bf\psi(\bx)\geq -\alpha_1(\psi(\bx))+\varepsilon
$$
must hold to avoid a contradiction. In addition, by construction of a ReCBF~\eqref{eq:ReCBF}, if $L_\bf \psi(\bx) \geq -\alpha_1(\psi(\bx))+\varepsilon$, then $h(\bx)=\psi(\bx)$ and $L_\bf h(\bx)=L_\bf\psi(\bx)$ (since $\Gamma$ and $\Gamma'$ evaluate to zero). Therefore,
$$
h(\bx)=0~\wedge~L_\bg h(\bx)=\bzero \implies L_\bf h(\bx)> 0,
$$
making it an OD-CBF according to Lemma~\ref{lem:OD-CBC_check}.
\end{proof}
Proposition~\ref{prop:OD-ReCBF} establishes that ReCBFs are also OD-CBFs if~\eqref{eq:OD-ReCBC} holds, and thus the results in Section~\ref{sec:OD-CBF} apply. Simulations for the double integrator~\eqref{sys:double} with the QP controllers using ReCBF and OD-ReCBF are shown in Fig.~\ref{fig:double_integrator}.

\section{Application to Satellite Control}
For a more realistic example, we apply the OD-HOCBF and OD-ReCBF frameworks to the safe control of a satellite. We consider the satellite dynamics on the orbital plane under Newton's gravitational model, in polar coordinates:
\begin{eqnarray}
\label{eqn:spacecraft}
\underbrace{ \begin{bmatrix}\dot r \\ \dot \theta  \\ \ddot r \\ \ddot \theta  \end{bmatrix} }_{\dot{\bx}} =  \underbrace{\begin{bmatrix}\dot r \\ \dot \theta  \\ r\dot \theta^2 - \mu r/(r^2)^{3/2} \\ -(2/r)\dot r \dot \theta \\   \end{bmatrix}}_{\bf(\bx)}+ \underbrace{\begin{bmatrix}
 0 \\ 0 \\ u_1 \\ u_2/r 
 \end{bmatrix}}_{\bg(\bx) \bu} ,
\end{eqnarray}
where ${\mu= 2.346\times10^{-9}}$ $\text{km}^3/\text{s}^2$ is the gravitational parameter for the central body. The radius of the central body is $R=0.3097$ km, and we are interested in maintaining the satellite within the range of $1.8R$ to $2.2R$. To this end, we use the safety constraint function:
$\psi(\bx) = 1-(r-2R)^2/(0.2R)^2$. We note that, like in the double integrator example, $\psi$ is not an HOCBF, due to the loss of relative degree where $r=2R$. This issue can be addressed using either the OD-HOCBF or the (OD-)ReCBF framework, the latter of which is developed specifically for this issue.

We use the QP~\eqref{eq:OD-CBF-QP} with $h$ produced by~\eqref{eq:HOCBF} and~\eqref{eq:ReCBF}. The relevant parameters are: ${\bk_\des (\bx) = \bzero}$, ${\Gamma = \bI}$, ${p=1}$, ${\alpha_1(s)=s/600}$, ${\alpha_2(s) = s/200}$, ${\theta_\des=1}$, ${\varepsilon=0.0141}$, and ${C_1=271.44}$. We simulate the satellite's trajectories from a random initial condition in the safe set 
$\bx_0=(0.6649,2.034, 2.346, 8.097)$
for a period of two days. The resulting trajectories are given in Fig.~\ref{fig:satellite}. 
Over
the two-day simulation timeframe, the OD-HOCBF framework produces a control signal that is smoother and much less aggressive compared to that of the OD-ReCBF. Yet, in both cases, safety is maintained under our developed frameworks.

\begin{figure}
    \centering
    \includegraphics[width=\linewidth]{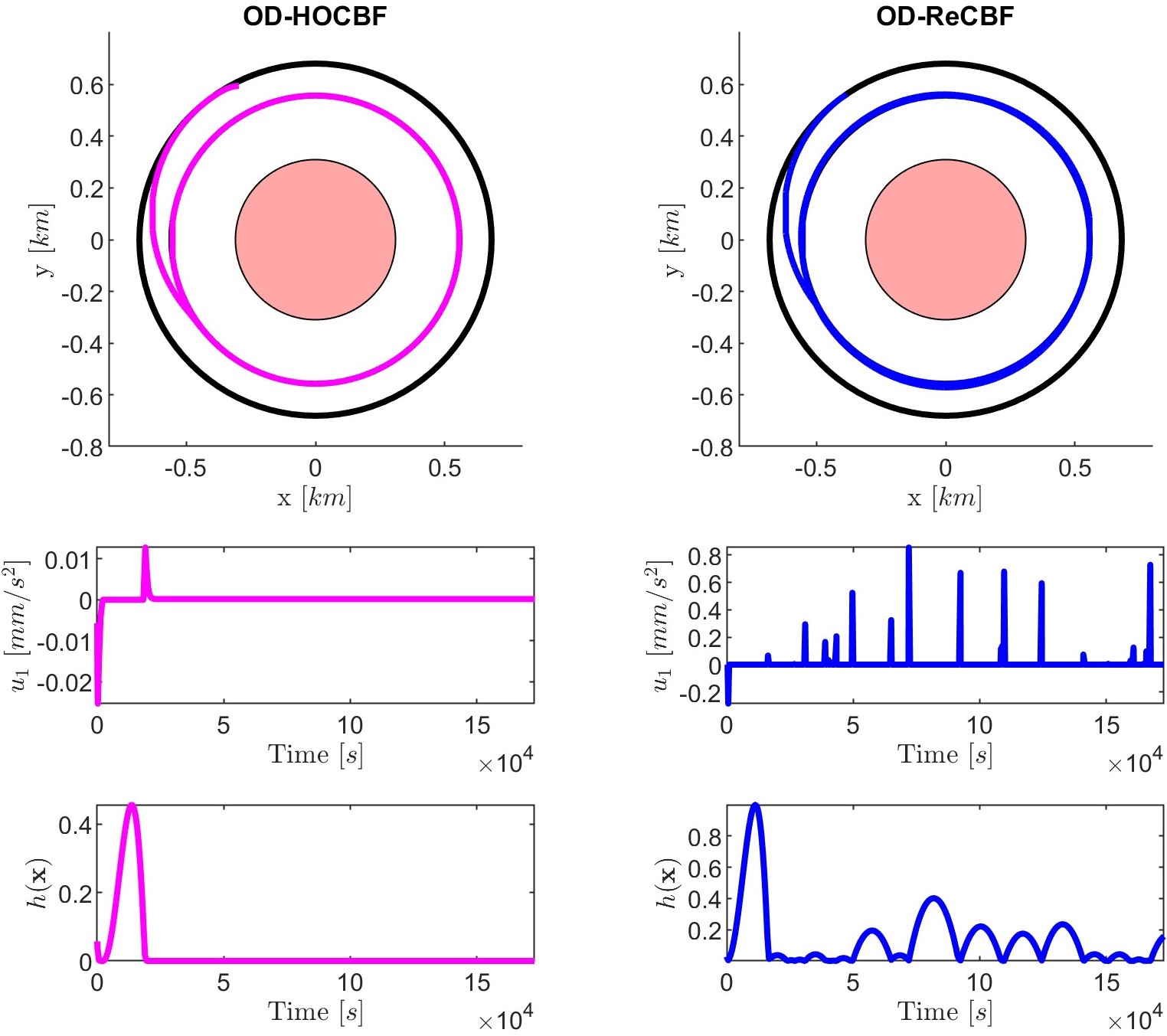}
    \vspace{-0.75cm}
    \caption{Satellite orbit over the period of two days. \textbf{Top:} State trajectory of the satellites for each optimal-decay method. \textbf{Middle:} Control signals generated from the QP controllers. \textbf{Bottom:} The function $h$ for each method that remains positive, certifying system safety.%
    \vspace{-0.75cm}
    } 
    \label{fig:satellite}
\end{figure}

\section{Conclusion}

We established a rigorous formalization of optimal-decay control barrier functions (OD-CBFs).
We derived conditions to verify the validity of OD-CBFs, established the feasibility and closed-form solution of quadratic program-based controllers, and proved the forward invariance and stability of the resulting safe set.
We highlighted the advantages of this framework in the context of higher-order CBF constructions, and showed that it provides beneficial properties for cases with vanishing relative degrees.
We demonstrated the applicability of this framework on a satellite control problem.

The main limitation of our results is that they rely on the assumption of an unbounded admissible input set, which is somewhat in contrast to the original motivation of OD-CBFs in \cite{AmesIEEEA20} \cite{ZengACC21}. Yet, as argued earlier, bounded inputs are not the only factor that make the search for a compatible class $\mathcal{K}^e$ function challenging, and the results presented herein provide a comprehensive understanding of how OD-CBFs address the pitfalls of CBFs with a vanishing relative degree, especially for higher order safety constraints. Future efforts will focus on extending our results to the bounded input case.

\bibliographystyle{ieeetr}
\bibliography{
    bib/alias,
    bib/PO,
    bib/main-Pio,
    bib/MC}
\end{document}